\documentclass{llncs}
\pdfoutput=1

\usepackage{paralist,amsmath,amssymb,url,graphicx}

\newcommand{\cobasa}[0]{\textsf{CoBaSA}}
\newcommand{\inez}[0]{\textsf{Inez}}
\newcommand{\comment}[1]{}
\newcommand{\ie}[0]{\emph{i.e.}, }
\newcommand{\eg}[0]{\emph{e.g.}, }

\newcommand\none[0]{\ensuremath{\mathtt{None}}}

\newcommand{\tr}[1]{\ensuremath{\mathtt{#1}}}
\newcommand{\ttr}[1]{\ensuremath{T\text{-}\mathtt{#1}}}

\newcommand{\state}[2]{\ensuremath{#1\ \|\ #2}}
\newcommand{\statenone}[1]{\state{#1}{\none}}
\newcommand{\subp}[2]{\ensuremath{\langle#1, #2\rangle}}

\newcommand{\Z}[0]{\ensuremath{\mathcal{Z}}}
\newcommand{\ZT}[0]{\ensuremath{\mathcal{Z} \cup T}}
\newcommand{\SigmaZ}[0]{\ensuremath{\Sigma_{\mathcal{Z}}}}

\newcommand{\SigmaZT}[0]{\ensuremath{\SigmaZ{} \cup \Sigma}}
\newcommand{\modelsz}[0]{\ensuremath{ \models_{\Z} }}

\newcommand{\modelszt}[0]{\ensuremath{ \models_{\Z \cup T} }}

\newcommand{\set}[2]{\ensuremath{\{#1\ |\ #2\}}}
\newcommand{\setsp}[2]{\ensuremath{\{ \subp{#1}{#2} \}}}

\newcommand{\trans}[0]{\ensuremath{ \longrightarrow }}
\newcommand{\transplus}[0]{\ensuremath{ \longrightarrow^{+} }}
\newcommand{\transstar}[0]{\ensuremath{ \longrightarrow^{*} }}
\newcommand{\transannot}[1]{\underset{\tr{#1}}{\longrightarrow}}

\newcommand{\cd}{\subp{C}{D}}
\newcommand{\cdprime}{\subp{C'}{D'}}

\newcommand{\pa}{\state{P}{A}}
\newcommand{\paprime}{\state{P'}{A'}}

\newcommand{\bct}[0]{\ensuremath{\operatorname{BC}(T)}}

\newcommand{\dpllt}[0]{DPLL($T$)}

\newtheorem{fact}{Fact}

\newcommand{\obj}[1]{\ensuremath{\operatorname{\mathsf{obj}}(#1)}}
\newcommand{\lb}[1]{\ensuremath{\operatorname{\mathsf{lb}}(#1)}}

\DeclareMathOperator{\aread}{\mathsf{read}}
\DeclareMathOperator{\awrite}{\mathsf{write}}

\DeclareMathOperator{\intf}{\mathsf{intf}}
\DeclareMathOperator{\intfz}{\mathsf{intf}_{\Z}}
\DeclareMathOperator{\intfsigma}{\mathsf{intf}_{\Sigma}}
\DeclareMathOperator{\maxc}{\mathsf{maxc}}
\DeclareMathOperator{\zatoms}{\mathsf{atoms_{\Z}}}
\DeclareMathOperator{\bounds}{\mathsf{bounds}}
\DeclareMathOperator{\zvars}{\mathsf{vars_{\Z}}}

\begin{document}

\author{Panagiotis Manolios and Vasilis Papavasileiou}

\institute{Northeastern University \\
  \email{\texttt{\{pete,vpap\}@ccs.neu.edu}}}

\title{ILP Modulo Theories}

\maketitle

\begin{abstract}
  We present Integer Linear Programming (ILP) Modulo Theories
  (IMT). An IMT instance is an Integer Linear Programming instance,
  where some symbols have interpretations in background theories.  In
  previous work, the IMT approach has been applied to industrial
  synthesis and design problems with real-time constraints arising in
  the development of the Boeing 787. Many other problems ranging from
  operations research to software verification routinely involve
  linear constraints and optimization. Thus, a general ILP Modulo
  Theories framework has the potential to be widely applicable.  The
  logical next step in the development of IMT and the main goal of
  this paper is to provide theoretical underpinnings.  This is
  accomplished by means of \bct{}, the Branch and Cut Modulo $T$
  abstract transition system. We show that \bct{} provides a sound and
  complete optimization procedure for the ILP Modulo $T$ problem, as
  long as $T$ is a decidable, stably-infinite theory. We compare a
  prototype of \bct{} against leading SMT solvers.
\end{abstract}

% \terms Design, Languages

% \keywords Decision Procedures, Satisfiability Modulo Theories,
% Linear Programming, Optimization

\section{Introduction}

The primary goal of this paper is to present the theoretical
underpinnings of the Integer Linear Programming (ILP) Modulo Theories
(IMT) framework for combining ILP with background theories. The
motivation for developing the IMT framework comes from our previous
work, where we used an ILP-based synthesis tool, \cobasa{}
(Component-Based System Assembly), to algorithmically synthesize
architectural models using the actual production design data and
constraints arising during the development of the Boeing 787
Dreamliner~\cite{hmp11}. According to Boeing engineers, previous
methods for creating architectural models required the
\emph{``cooperation of multiple teams of engineers working over long
  periods of time.''} We were able to synthesize architectures in
minutes, directly from the high-level requirements. What made this
possible was the combination of ILP with a custom decision procedure
for hard real-time constraints~\cite{hmp11}, \ie an instance of IMT.

ILP has been the subject of intensive research for more than five
decades~\cite{gomory63}. ILP solvers~\cite{cplex,scip} are routinely
used to solve practical optimization problems from a diverse set of
fields including operations research, industrial engineering,
artificial intelligence, economics, and software verification. Based
on our successful use of the IMT approach to solve architectural
synthesis problems and the widespread applicability of ILP and
optimization, we hypothesize that IMT has the potential to enable
interesting new applications, analogous to what is currently happening
with Satisfiability Modulo
Theories~\cite{barrett02,demoura02,dpllt,z3}.

We introduce the theoretical underpinnings of IMT via the \bct{}
framework (Branch and Cut Modulo $T$). \bct{} can be thought of as the
IMT counterpart to the \dpllt{} architecture for lazy
SMT~\cite{dpllt}. \bct{} models the branch-and-cut family of
algorithms for integer programming as an abstract transition system
and allows plugging in theory solvers. Building on classical results
on combining decision procedures~\cite{no79,noproof,combiningdp}, we
show that \bct{} provides a sound and complete optimization procedure
for the combination of ILP with stably-infinite theories. As a
side-product of our theoretical study of IMT, we show how to bound
variables while preserving optimality modulo the combination of Linear
Integer Arithmetic and a stably-infinite theory.

The rest of the paper is organized as follows. In
Section~\ref{sec:bct}, we formally define IMT and provide an abstract
\bct{} architecture for solving IMT problems. IMT can be seen as SMT
with a more expressive core than propositional logic.  We elaborate on
the relationship between IMT and SMT in Section~\ref{sec:smt}.  We
have implemented \bct{}, using the SCIP MIP solver~\cite{scip} as the
core solver. We carried out a sequence of experiments, as outlined in
Section~\ref{sec:experiments}. The first experiment shows that for our
synthesis problems, ILP solvers~\cite{cplex,scip} outperform the Z3
SMT solver~\cite{z3}.  In the second experiment, we compared our
prototype implementation with state-of-the-art SMT
solvers~\cite{z3,mathsatlia} on SMT-LIB benchmarks. An analysis of the
results suggests that \bct{} is an interesting future alternative to
the \dpllt{} architecture.  We provide an overview of related work in
Section~\ref{sec:related} and conclude with
Section~\ref{sec:conclusions}.

\section{\bct{}}
\label{sec:bct}

In this section, we formally define IMT. We also provide a general
\bct{} architecture for solving IMT problems. We describe \bct{} by
means of a \emph{transition system}, similar in spirit to
\dpllt{}~\cite{dpllt}. The \bct{} architecture allows one to obtain a
solver for ILP Modulo $T$ by combining a branch-and-cut ILP solver
with a background solver for $T$.

\subsection{Formal Preliminaries}
\label{sec:def}

\comment{We assume a fixed set of variable symbols $\mathcal{V}$.}

An \emph{integer linear expression} is a sum of the form $c_1 v_1 +
\cdots + c_n v_n$ for integer constants $c_i$ and variable symbols
$v_i$. An \emph{integer linear constraint} is a constraint of the form
$e\ \bowtie\ r$, where $e$ is an integer linear expression, $r$ is an
integer constant, and $\bowtie$ is one of the relations $<$, $\leq$,
$=$, $>$, and $\geq$. An \emph{integer linear formula} is a set of
(implicitly conjoined) integer linear constraints. We will use
propositional connectives over integer linear constraints and formulas
as appropriate and omit $\wedge$ when this does not cause ambiguity
(\ie juxtaposition will denote conjunction). An \emph{integer linear
  programming (ILP) instance} is a pair $C, O$, where $C$ is an
integer linear formula, and the \emph{objective function} $O$ is an
integer linear expression. Our goal will always be \emph{minimizing}
the objective function.

We assume a fixed set of variables $\mathcal{V}$. An \emph{integer
  assignment} $A$ is a function $\mathcal{V} \rightarrow \mathbb{Z}$,
where $\mathbb{Z}$ is the set of integers. We say that an assignment
$A$ \emph{satisfies} the constraint $c = (c_1 v_1 + \cdots + c_n v_n\
\bowtie\ r)$ (where $\bowtie$ is one of the relations $<$, $\leq$,
$=$, $>$, $\geq$, and every $v_i$ is in $\mathcal{V}$) if $\sum_i c_i
\cdot A(v_i)\ \bowtie\ r$.  An assignment $A$ satisfies a formula $C$
if it satisfies every constraint $c \in C$.  A formula $C$ is
\emph{integer-satisfiable} or \emph{integer-consistent} if there is an
assignment $A$ that satisfies $C$. Otherwise, it is called
\emph{integer-unsatisfiable} or \emph{integer-inconsistent}.

A signature $\Sigma$ consists of a set $\Sigma^C$ of constant symbols,
a set $\Sigma^F$ of function symbols, a set $\Sigma^P$ of predicate
symbols, and a function $ar: \Sigma^F \cup \Sigma^P \rightarrow
\mathbb{N}^{+}$ that assigns a non-zero natural number (the arity) to
every function and predicate symbol. A $\Sigma$-formula is a
first-order logic formula constructed using the symbols in $\Sigma$.
A $\Sigma$-theory $T$ is a closed set of $\Sigma$-formulas (\ie $T$
contains no free variables). We will write theory in place of
$\Sigma$-theory when $\Sigma$ is clear from the context (similarly for
terms and formulas).

\begin{example}
  Let $\Sigma_{\mathcal{A}}$ be a signature that contains a binary
  function $\aread$, a ternary function $\awrite$, no constants, and
  no predicate symbols. The theory $T_{\mathcal{A}}$ of arrays
  (without extensionality) is defined by the following
  formulas~\cite{mc62}:
  \begin{align*}
    \forall a\ \forall i\ \forall e\ & [\aread(\awrite(a, i, e), i)
    = e] \\
    \forall a\ \forall i\ \forall j\ \forall e\ & [i \neq j \Rightarrow
    \aread(\awrite(a, i, e), j) = \aread(a, j)].
  \end{align*}
\end{example}

A formula $F$ is $T$-satisfiable or $T$-consistent if $F \wedge T$ is
satisfiable in the first-order sense (\ie there is an interpretation
that satisfies it). A formula $F$ is called $T$-unsatisfiable or
$T$-inconsistent if it is not $T$-satisfiable. For formulas $F$ and
$G$, $F$ $T$-entails $G$ (in symbols $F \models_T G$) if $F \wedge
\neg G$ is $T$-inconsistent.

\begin{definition}
  Let \SigmaZ{} be a signature that contains the constant symbols $\{
  0,$ $\pm 1,$ $\pm 2,$ $\ldots \}$, a binary function symbol $+$, a unary
  function symbol $-$, and a binary predicate symbol $\leq$. The
  theory of Linear Integer Arithmetic, which we will denote by $\Z$,
  is the \SigmaZ{}-theory defined by the set of closed
  \SigmaZ{}-formulas that are true in the standard model (an
  interpretation whose domain is $\mathbb{Z}$, in which the symbols in
  \SigmaZ{} are interpreted according to their standard meaning over
  \Z{}).
\end{definition}

We will use relation symbols like $<$ that do not appear in \SigmaZ{},
and also multiplication by a constant (which is to be interpreted as
repeated addition); these are only syntactic shorthands. We will
frequently view an integer assignment $A$ as the set of formulas
$\set{v = A(v)}{v \in \mathcal{V}}$, where $A(v)$ is viewed as a
\SigmaZ{}-term. An integer assignment $A$ viewed as a set of formulas
is always $\Z$-consistent. If $A$ is an integer assignment and $A$
satisfies an integer linear formula $C$, it is also the case that $A
\modelsz C$.  If $A$ is an integer assignment, $T$ is a theory and $F$
is a formula, we will say that $A$ is a $T$-model of $F$ if $A$ is
$T$-consistent and $A \modelszt F$. Note that a $T$-model is not a
first-order model.

\comment{Whenever convenient, we will be viewing first-order models as
  integer assignments; this is possible, because for every first-order
  model $M$ of \Z{} and for every variable symbol $v$, there is a
  unique integer constant $c$ such that $M \models v = c$.}

A \emph{$\Sigma$-interface atom} is a $\Sigma$-atomic formula (\ie the
application of a predicate symbol or equality), possibly annotated
with a variable symbol, \eg $(x=y)^v$. The meaning of a
$\Sigma$-interface atom with no annotation remains the same. An
annotated $\Sigma$-interface atom $\phi^v$ denotes $\phi
\Leftrightarrow v > 0$. A set of $\Sigma$-interface atoms will often
be used to denote their conjunction. \comment{We will assume that for
  every interface atom of the form $x = y$ or $(x = y)^v$ where $x$,
  $y$, and $v$ are variable symbols, at most one of $x$ and $y$
  appears in integer linear constraints; the opposite would defy the
  purpose of interface atoms.}

\begin{definition}[ILP Modulo $T$ Instance]
  An \emph{ILP Modulo (Theory) $T$ instance}, where the signature
  $\Sigma$ of $T$ is disjoint from \SigmaZ{}, is a triple of the form
  $C, I, O$, where $C$ is an integer linear formula, $I$ is a set of
  $\Sigma$-interface atoms, and $O$ is an objective function. The
  variables that appear in both $C$ and $I$ are called interface
  variables.
  \label{def:imtinstance}
\end{definition}

An ILP Modulo $T$ instance can be thought of as an integer linear
program that contains terms which have meaning in $T$. In
Definition~\ref{def:imtinstance}, the interface atoms (elements of
$I$) are separated from the linear constraints, \ie there are no
$\Sigma$-terms embedded within integer linear constraints.  This is
not a restriction, because every set of $(\Sigma \cup
\SigmaZ{})$-atomic formulas can be written in separate
form~\cite[``Variable Abstraction'']{combiningdp}.

\begin{example}
  Let $\Sigma$ be a signature that contains the unary function symbol
  $f$. The formula $f(x + 1) + f(y + 2) \geq 3$ (where $x$ and $y$ are
  variable symbols) can be written in separate form as $C = \{ v_3 +
  v_4 \geq 3, v_1 = x + 1, v_2 = y + 2 \}$ and $I = \{ v_3 =
  f(v_1), v_4 = f(v_2) \}$. $C$ is an integer linear formula; $I$ is a
  set of $\Sigma$-interface atoms; and $\Sigma$ is disjoint from
  \SigmaZ{}. Variable abstraction introduced new variables, $v_{1},
  \ldots, v_4$. $C$ and $I$ only share variable symbols.

  Let $A$ be the assignment $\{x = 2, y = 1, v_1 = 3, v_2 = 3, v_3 =
  3, v_4 = 0\}$. Clearly $A \modelsz C$. However, $A \not
  \models_\emptyset I$, where $\emptyset$ stands for the \emph{theory
    of uninterpreted functions} (also called the empty theory, because
  it has an empty set of formulas). The reason is that $v_1 = v_2$ but
  $f(v_1) \neq f(v_2)$. In contrast, the assignment $A' = \{x = 2, y =
  1, v_1 = 3, v_2 = 3, v_3 = 3, v_4 = 3\}$ is a $\emptyset$-model of
  $C \wedge I$ per our definition, as it is $\emptyset$-consistent and
  $A' \models_{\Z{} \cup \emptyset} C \wedge I$.
\label{ex:abstraction}
\end{example}

\subsection{Transition System}

\begin{definition}[Difference Constraint]
  A difference constraint is a constraint of the form $v_i \leq c$ or
  $v_i - v_j \leq c$, where $v_i$ and $v_j$ are integer variables and
  $c$ is an integer constant.
\end{definition}

\begin{definition}[Subproblem]
  A subproblem is a pair of the form \cd{}, where $C$ is a set of
  constraints and $D$ is a set of difference constraints.
\end{definition}

In a subproblem \cd{}, we distinguish between the arbitrary
constraints in $C$ and the simpler constraints in $D$ in order to
provide a good interface for the interaction between the core ILP
solver and background theory solvers that only understand difference
logic, \ie a limited fragment of $\Z$. It is the responsibility of the
core solver to notify the theory solver about the difference
constraints that hold. Difference constraints are clearly a special
case of integer linear constraints.

\begin{definition}[State]
  A state of \bct{} is a tuple {\pa}, where $P$ is a set of
  subproblems, and $A$ is either the constant \none, or an assignment.
  If $A$ is an assignment, it can optionally be annotated with the
  superscript $-\infty$.
\end{definition}

Our abstract framework maintains a list of open subproblems, because
it is designed to allow different branching strategies. This is in
contrast to an algorithm like CDCL that does not keep track of
subproblems explicitly. There, subproblems are implicit, \ie
backtracking can reconstruct them. ILP solvers branch over non-Boolean
variables in arbitrary ways, thus mandating that we explicitly record
subproblems.

In a state \pa{}, the assignment $A$ is the best known
($T$-consistent) solution so far, if any. It has a superscript
$-\infty$ if it satisfies all the constraints, but is not optimal
because the IMT instance admits solutions with arbitrarily low
objective values. If this is the case, it is useful to provide an
assignment and to also report that no optimal assignment exists.

The interface atoms $I$ and the objective function $O$ are not part of
the \bct{} states because they do not change over time.  $\obj{A}$
denotes the value of the objective function $O$ under assignment $A$:
if $O = \sum_i c_i v_i$, then $\obj{A} = \sum_i c_i \cdot A(v_i)$. The
objective function itself is not an argument to
$\operatorname{\mathsf{obj}}$ because it will be clear from the
context which objective function we are referring to. For convenience,
we define $\obj{\none} = +\infty$ and $\obj{A^{ -\infty}} =
-\infty$. Function $\lb{\cd}$ returns a lower bound for the possible
values of the objective function $O$ for the subproblem \cd{}: by
definition, there is no $A$ such that $A$ satisfies $C \wedge D$ and
$\obj{A} < \lb{\cd}$.

Figure~\ref{fig:bct-trans} defines the \emph{transition relation}
$\longrightarrow$ of \bct{} (a binary relation over states). In the
rules, $c$ and $d$ always denote integer linear constraints and
difference constraints. $C$ (possibly subscripted) denotes an integer
linear formula (set of integer linear constraints), while $D$ denotes
a set of difference constraints. $C\ c$ stands for the set union $C
\cup \{c\}$, under the implicit assumption that $c \notin C$;
similarly for $D\ d$. $C$ and $D$ are always well-formed sets, \ie
they contain no syntactically duplicate elements. $P$ and $P'$ stand
for sets of syntactically distinct subproblems, while $A$ and $A'$ are
integer assignments.  $P \uplus Q$ denotes the union $P \cup Q$, under
the implicit assumption that the two sets are disjoint. \comment{$F$
  stands for a $(\Sigma \cup \Sigma_\Z)$-formula, where $\Sigma$ is
  the signature of $T$.} The intuitive meaning of the different \bct{}
rules is the following:

\begin{figure}[t!]
  \begin{align*}
    \tr{Branch}\ \ &
    \begin{aligned}
      & \state{P \uplus \{ \cd \}}{A} \trans
      \state{P \cup \{\subp{C_i}{D}\ |\ 1 \leq i \leq n\}}{A} \\
      & \text{if } \left\{
        \begin{array}{l}
          n > 1 \\
          D \modelsz (C \Leftrightarrow \bigvee_{1 \leq i \leq n}
          C_i) \\
          C_i \text{ are syntactically distinct}
        \end{array}
      \right.
      \\
    \end{aligned} \\[8pt]
    \tr{Learn}\ \ &
    \begin{aligned}
      & \state{P \uplus \{ \cd \}}{A} \trans
      \state{P \cup \{ \subp{C\ c}{D} \}}{A} \\
      & \text{if $C \wedge D \modelsz c$}
    \end{aligned} \\[8pt]
    \ttr{Learn}\ \ &
    \begin{aligned}
      & \state{P \uplus \{ \cd \}}{A} \trans
      \state{P \cup \{ \subp{C\ c}{D} \} }{A} \\
      & \text{if $C \wedge D \wedge I \modelszt c$}
    \end{aligned} \\[8pt]
    \tr{Propagate}\ \ &
    \begin{aligned}
      & \state{P \uplus \{ \cd \}}{A} \trans
      \state{P \cup \{ \subp{C}{D\ d} \}}{A} \\
      & \text{if $C \wedge D \modelsz d$}
    \end{aligned} \\[8pt]
    \tr{Forget}\ \ &
    \begin{aligned}
      & \state{P \uplus \{ \subp{C\ c}{D} \}}{A} \trans
      \state{P \cup \{ \subp{C}{D} \}}{A} \\
      & \text{if $C \wedge D \modelsz c$}
    \end{aligned} \\[8pt]
    \tr{Drop}\ \ &
    \begin{aligned}
      & \state{P \uplus \{ \cd \}}{A} \trans \pa \\
      & \text{if $C \wedge D$ is integer-inconsistent}
    \end{aligned} \\[8pt]
    \tr{Prune}\ \ &
    \begin{aligned}
      & \state{P \uplus \{ \cd \}}{A} \trans \pa \\
      & \text{if }
      \left\{
        \begin{array}{l}
          A \neq \none \\
          \lb{\cd} \geq \obj{A}
        \end{array}
      \right.
    \end{aligned} \\[8pt]
    \tr{Retire}\ \ &
    \begin{aligned}
      & \state{P \uplus \{ \cd \}}{A} \trans \state{P}{A'} \\
      & \text{if }
      \left\{
        \begin{array}{l}
          A' \text{ is a $T$-model of } C \wedge D \wedge I \\
          \obj{A'} < \obj{A} \\
          \text{for any $T$-model $B$ of $C \wedge D
            \wedge I$, $\obj{A'} \leq \obj{B}$} \\
          % \text{\hspace{20pt}then $\obj{A'} \leq \obj{B}$} \\
        \end{array}
      \right.
    \end{aligned} \\[8pt]
    \tr{Unbounded}\ \ &
    \begin{aligned}
      & \state{P \uplus \{ \cd \}}{A} \trans
      \state{\emptyset}{A'^{\ -\infty}} \\
      & \text{if } \left\{
        \begin{array}{l}
          A' \text{ is a $T$-model of } C \wedge D \wedge I \\
          \obj{A'} \leq \obj{A} \\
          \text{for any integer $k$, there exists a
            $T$-model $B$ of $C \wedge D \wedge I$} \\
          \hspace{6pt}\text{such that $\obj{B} < k$}
        \end{array}
      \right.
    \end{aligned}
  \end{align*}
  \caption{The \bct{} Transition System}
  \label{fig:bct-trans}
\end{figure}

\begin{description}
\item[\tr{Branch}] \hfill \\
  Case-split on a subproblem \cd{}, by replacing it with two or more
  different subproblems $\subp{C_i}{D}$. If there is a satisfying
  assignment for $C \wedge D$, this assignment will also satisfy $C_i
  \wedge D$ for some $i$, and conversely.
\item[\tr{Learn}, \ttr{Learn}, \tr{Propagate}] \hfill \\
  Add an entailed constraint (in the case of \tr{Learn} and
  \ttr{Learn}) or difference constraint (\tr{Propagate}) to a
  subproblem. \ttr{Learn} takes the theory $T$ into
  account. \ttr{Learn} is strictly more powerful than \tr{Learn}. We
  retain the latter as a way to denote transitions that do not involve
  theory reasoning.
\item[\tr{Forget}] \hfill \\
  Remove a constraint entailed by the remaining constraints of a
  subproblem.
\item[\tr{Drop}, \tr{Prune}] \hfill \\
  Eliminate a subproblem either because it is unsatisfiable
  (\tr{Drop}), or because it cannot lead to a solution better than the
  one already known.
\item[\tr{Retire}, \tr{Unbounded}] \hfill \\
  The solution to a subproblem becomes the new incumbent solution, as
  long as it improves upon the objective value of the previous
  solution. If there are solutions with arbitrarily low objective
  values, we don't need to consider other subproblems.
\end{description}

The observant reader will have noticed that the \ttr{Learn} rule is
very powerful, \ie it allows for combined \ZT{}-entailment. This is in
pursuit of generality. Our completeness strategy
(Theorem~\ref{thm:complete}) will not depend in any way on performing
combined arithmetic and theory reasoning, but only on extracting
equalities and disequalities from the difference
constraints. Entailment modulo \ZT{} is required if we want to learn
clauses, because they are represented as linear
constraints. Interesting implementations of \bct{} may go beyond
clauses and apply \ttr{Learn} for theory-specific cuts.

We define the binary relations $\transplus$ and $\transstar$ over
\bct{} states as follows: $S \transplus S'$ if $S \trans S'$, or there
exists some state $Q$ such that $S \transplus Q$ and $Q \trans S'$. $S
\transstar S'$ if $S = S'$ or $S \transplus S'$. When convenient, we
will annotate a transition arrow between two \bct{} states with the
name of the rule that relates them, for example $S \transannot{Branch}
S'$.

A \emph{starting state} for \bct{} is a state of the form
$\statenone{\setsp{C}{\emptyset}}$, where $C$ is the set of integer
linear constraints of an ILP Modulo $T$ instance. A \emph{final state}
is a state of the form \state{\emptyset}{A} ($A$ can also be \none{},
or an assignment annotated with $-\infty$).

\subsection{Soundness and Completeness}
\label{subsec:soundness}

Throughout this Section, we assume an IMT instance with objective
function $O$ and a set of interface atoms
$I$. Theorems~\ref{thm:sound-sat} and~\ref{thm:sound-unsat}
characterize \bct{} soundness. Theorem~\ref{thm:complete} relies upon
classical results for combining decision
procedures~\cite{no79,noproof,combiningdp} and shows that~\bct{} can
be applied in a complete way.

\begin{lemma}
  \label{lemma:l1}
  For states \pa{} and \paprime{} such that $\pa \trans \paprime$, if
  there is an assignment $B$ such that $B$ is a $T$-model of $C \wedge
  D \wedge I$ for some $\cd \in P$, then one of the following
  conditions holds: either
  \begin{inparaenum}[(i)]
  \item\label{cond:l1:i} $\obj{A'} \leq \obj{B}$, or
  \item\label{cond:l1:ii} $B$ is a $T$-model of $C' \wedge D' \wedge
    I$ for some $\cdprime \in P'$.
  \end{inparaenum}
\end{lemma}
\begin{proof}
  Assume that there is a subproblem $p = \cd \in P$ and an assignment
  $B$ such that $B$ is a $T$-model of $C \wedge D \wedge I$. If $p \in
  P'$, then~(\ref{cond:l1:ii}) holds trivially. If $p \notin P'$, then
  the transition cannot possibly be \tr{Drop} (we cannot apply
  \tr{Drop} on {\cd} because $B$ is an assignment that satisfies $C
  \wedge D$). For the other rules:
  \begin{itemize}
  \item \tr{Branch}: There are subproblems $\subp{C_1}{D}$, $\ldots$,
    $\subp{C_n}{D}$ in $P'$ such that $D \modelsz (C \Leftrightarrow
    \bigvee_{1 \leq i \leq n} C_i)$.  Thus, $B$ is a $T$-model of $C_i
    \wedge D \wedge I$ for some $i$ ($1 \leq i \leq
    n$). $\subp{C_i}{D} \in P'$, therefore~(\ref{cond:l1:ii}) holds.
  \item \tr{Prune}: $\obj{A'} \leq \lb{C \wedge D} \leq
    \obj{B}$. Condition~(\ref{cond:l1:i}) holds.
  \item \tr{Retire}, \tr{Unbounded}: $\obj{A'} \leq \obj{B}$,
    therefore~(\ref{cond:l1:i}) holds.
  \item \tr{Learn}, \tr{Forget}, \tr{Propagate}, \ttr{Learn}: there is
    a subproblem $\subp{C'}{D'} \in P'$ such that $C \wedge D \wedge I
    \modelszt C' \wedge D'$, and therefore $B$ is a $T$-model of $C'
    \wedge D' \wedge I$.
  \end{itemize}
\end{proof}

\begin{lemma}
  \label{lemma:assignlt}
  For states {\pa} and {\paprime} such that $\pa \trans \paprime$,
  $\obj{A'} \leq \obj{A}$.
\end{lemma}
\begin{proof}
  The only rules that modify the assignment are \tr{Retire} and
  \tr{Unbounded}; the conditions under which we can apply them imply
  $\obj{A'} \leq \obj{A}$.  For any other rule, $A = A'$.
\end{proof}

\begin{lemma}
  \label{lemma:assigndiff}
  For states {\pa} and {\paprime} such that $\pa \trans \paprime$, if
  $A \neq A'$ then $A'$ is a $T$-model of $C \wedge D \wedge I$ for
  some $\cd{} \in P$.
\end{lemma}
\begin{proof}
  The conditions of \tr{Retire} and \tr{Unbounded} guarantee that $A'$
  is a $T$-model of $C \wedge D \wedge I$ for some $\cd \in P$. No
  other rule modifies the assignment.
\end{proof}

\begin{lemma}
  \label{lemma:l2}
  For states {\pa} and {\paprime} such that $\pa \transstar \paprime$,
  if there is an assignment $B$ such that $B$ is a $T$-model of $C
  \wedge D \wedge I$ for some $\cd \in P$, then one of the following
  conditions holds: $\obj{A'} \leq \obj{B}$, or $B$ is a $T$-model of
  $C' \wedge D' \wedge I$ for some $\cdprime \in P'$.
\end{lemma}
\begin{proof}
  We induct on the length $n$ of the sequence of transitions.

  \begin{itemize}
  \item Induction base: $n = 0$. $\pa = \paprime$; obvious.
  \item Induction step: assume that the property holds for any
    sequence of $n - 1$ transitions, where $n \geq 1$. We will prove
    that it holds for any sequence of transitions $\state{P_0}{A_0}
    \trans ... \trans \state{P_{n-1}}{A_{n-1}} \trans
    \state{P_n}{A_n}$.  Assume there is an assignment $B$ such that
    $B$ is a $T$-model of $C_0 \wedge D_0 \wedge I$ for some
    $\subp{C_0}{D_0} \in P_0$. By the induction hypothesis, one of the
    two following conditions holds:
    \begin{itemize}
    \item $\obj{A_{n - 1}} \leq \obj{B}$: then $\obj{A_n} \leq
      \obj{A_{n - 1}} \leq \obj{B}$ from
      Lemma~\ref{lemma:assignlt}.
    \item $B$ is a $T$-model of $C_{n - 1} \wedge D_{n - 1} \wedge I$
      for some subproblem $\subp{C_{n - 1}}{D_{n - 1}} \in P_{n - 1}$:
      our proof obligation follows from Lemma~\ref{lemma:l1} applied
      to the transition $\state{P_{n-1}}{A_{n-1}} \trans
      \state{P_n}{A_n}$.
    \end{itemize}
  \end{itemize}

\end{proof}

\begin{lemma}
  \label{lemma:entails}
  For states {\pa} and {\paprime} such that $\pa \trans \paprime$,
  $\bigvee_{\cd \in P'} C \wedge D \modelsz \bigvee_{\cd \in P} C
  \wedge D$.
\end{lemma}
\begin{proof}
  We case-split on the \bct{} transitions.
  \begin{itemize}
  \item The rules \tr{Prune}, \tr{Drop}, and \tr{Retire} can only make
    the disjunction of the subproblems in $P'$ stronger, because a
    subproblem is eliminated and the rest of the subproblems remain
    intact.
  \item For \tr{Unbounded}, $\bigvee_{\cd \in P'} C \wedge D \equiv
    \mathtt{false}$; $\mathtt{false} \modelsz \bigvee_{\cd \in P} C
    \wedge D$.
  \item The rules \tr{Learn}, \tr{Forget}, and \tr{Propagate} and
    substitute a subproblem for a $\Z$-equivalent one.
  \item The rule \ttr{Learn} adds a constraint to a subproblem, and
    therefore makes the disjunction in $P'$ stronger.
  \item The rule \tr{Branch} replaces a subproblem \cd{} with a set of
    subproblems whose disjunction is $\Z$-equivalent to $C \wedge D$.
  \end{itemize}
\end{proof}

\begin{lemma}
  \label{lemma:entailsstar}
  For states {\pa} and {\paprime} such that $\pa \transstar \paprime$,
  $\bigvee_{\cd \in P'} C \wedge D \modelsz \bigvee_{\cd \in P} C
  \wedge D$.
\end{lemma}
\begin{proof}
  Induction on the length of the sequence of transitions and
  application of Lemma~\ref{lemma:entails}.
\end{proof}

\begin{theorem}
  \label{thm:sound-sat}
  For a formula $C$ and an assignment $A$, if
  $$\statenone{\setsp{C}{\emptyset}} \transstar \state{\emptyset}{A}$$
  where $A \neq \none$, then
  \begin{inparaenum}[(a)]
  \item $A$ is a $T$-model of $C \wedge I$, and
  \item there is no assignment $B$ such that $B$ is a $T$-model of $C
    \wedge I$ and $\obj{B} < \obj{A}$.
  \end{inparaenum}
\end{theorem}
\begin{proof}
  \begin{enumerate}[(a)]

  \item The sequence of transitions from
    \statenone{\setsp{C}{\emptyset}} to \state{\emptyset}{A} has to be
    of the following form:
    
    \[
    \statenone{\setsp{C}{\emptyset}} \transstar \state{S_1}{A_1}
    \transannot{Retire / Unbounded}
    \state{S}{A} \transstar \state{\emptyset}{A}
    \]
    
    There is at least one \tr{Retire} or \tr{Unbounded} step, as these
    are the only rules that can modify the assignment. Consider the
    last such step. The conditions on \tr{Retire} and \tr{Unbounded}
    steps require that $A$ is a $T$-model of $I \wedge \bigvee_{\cd
      \in S_1} (C \wedge D)$. From Lemma~\ref{lemma:entailsstar},
    $\bigvee_{\cd \in S_1} (C \wedge D) \modelsz C$. Therefore, $A$
    is a $T$-model of $C$. \vspace{6pt}

  \item Follows from Lemma~\ref{lemma:l2}.
  
  \end{enumerate}
\end{proof}

\begin{theorem}
  \label{thm:sound-unsat}
  For a formula $C$, if $\statenone{\setsp{C}{\emptyset}} \transstar
  \statenone{\emptyset}$, then $C \wedge I$ is \ZT{}-unsatisfiable.
\end{theorem}
\begin{proof}
  Assume that there is an assignment $A$ such that $A$ is a $T$-model
  of $C$. Then, from Lemma~\ref{lemma:l2} either there exists
  $\subp{C'}{D'} \in \emptyset$ such that $A$ is a $T$-model of $C'
  \wedge D' \wedge I$ (which cannot possibly be true), or $+\infty =
  \obj{\none} \leq \obj{A}$ (contradiction, because $\obj{A}$ has to
  be finite).
\end{proof}

\begin{definition}[Stably-Infinite Theory]
  A $\Sigma$-theory $T$ is called stably-infinite if for every
   $T$-satisfiable quantifier-free $\Sigma$-formula $F$ there exists an
   interpretation satisfying $F \wedge T$ whose domain is infinite.
 \end{definition}

 \begin{definition}[Arrangement]
   Let $E$ be an equivalence relation over a set of variables $V$. The
   set
   \[
     \alpha(V, E) =\ \set{x = y}{x E y}\ \cup\ \set{x \neq y}{x, y \in
       V \text{ and not } x E y}
     \]
   is the arrangement of $V$ induced by $E$.
\end{definition}

Note that $\Z$ is a stably-infinite theory. We build upon the
following result on the combination of signature-disjoint
stably-infinite theories:

\begin{fact}[Combination of Stably-Infinite
  Theories~\cite{no79,noproof,combiningdp}]
  Let $T_i$ be a stably-infinite $\Sigma_i$-theory, for $i = 1, 2$,
  and let $\Sigma_1 \cap \Sigma_2 = \emptyset$. Also, let $\Gamma_i$
  be a conjunction of $\Sigma_i$ literals. $\Gamma_1 \cup \Gamma_2$ is
  $(T_1 \cup T_2)$-satisfiable iff there exists an equivalence
  relation $E$ of the variables shared by $\Gamma_1$ and $\Gamma_2$
  such that $\Gamma_i \cup \alpha(V, E)$ is $T_i$-satisfiable, for $i
  = 1, 2$.
  \label{thm:no}
\end{fact}

Decidability for the combination of $T_1 = \Z$ and another
stably-infinite theory follows trivially, as we can pick an
arrangement over the variables shared by the two sets of literals
non-deterministically and perform two $T_i$-satisfiability checks.

\begin{theorem}[Completeness]
  \label{thm:complete}
  \bct{} provides a complete optimization procedure for the ILP Modulo
  $T$ problem, where $T$ is a decidable stably-infinite theory.
\end{theorem}
\begin{proof}[Sketch]
  Let $\mathcal{C}, I, O$ be an ILP Modulo $T$ instance. Assume that
  $$\statenone{\setsp{\mathcal{C}}{\emptyset}} \transstar \pa\text{,}$$
  and that for every $\cd \in P$ the following conditions hold:
  \begin{inparaenum}[(a)]
  \item there is an equivalence relation $E_D$ over the set of
    interface variables $V$ of the ILP Modulo $T$ instance, such that
    $D \modelsz \alpha(V, E_D)$, and
  \item either $D \modelsz v > 0$ or $D \modelsz v \leq 0$ for every
    $v$ that appears as the annotation of an interface atom in $I$.
  \end{inparaenum}
  Then we can solve the IMT instance to optimality as follows. For
  every subproblem $\cd \in P$, ~$C \wedge D \wedge I$~ \ZT{}-entails
  the following set of literals:
  \begin{align*}
    \set{\phi}{\phi \in I \text{ and } \phi \text{ is not annotated}}
    &\ \cup \\
    \set{\phi}{\phi^v \in I \text{ and } D \modelsz v > 0}
    &\ \cup \\
    \set{\neg \phi}{\phi^v \in I \text{ and } D \modelsz v \leq 0}
    &\ \cup \\
    \alpha(V, E_D)
  \end{align*}
  If the set of literals is $T$-unsatisfiable, then $C \wedge D \wedge
  I$ is \ZT{}-unsatisfiable. If it is $T$-satisfiable, any integer
  solution for $C \wedge D$ will be a $T$-model. For the
  $T$-unsatisfiable subproblems, we apply \ttr{Learn} to learn an
  integer-infeasible constraint (\eg $0 < 0$) and subsequently apply
  \tr{Drop}. If all the subproblems are $T$-unsatisfiable, we reach a
  final state \state{\emptyset}{A}. If there are $T$-satisfiable
  subproblems, it suffices to let a (complete) branch-and-cut ILP
  algorithm run to optimality, as we have already established
  $T$-consistency. The basic steps of such algorithms can be described
  by means of \bct{} steps. Note that unbounded objective functions do
  not hinder completeness: it suffices to recognize an unbounded
  subproblem~\cite{unboundedip} and apply \tr{Unbounded}.

  A systematic branching strategy can guarantee that after a finite
  number of steps, the difference constraints of every subproblem
  entail an arrangement. For every pair of interface variables $x$ and
  $y$ and every subproblem, we apply the \tr{Branch} rule to obtain
  three new subproblems, each of which contains one of the constraints
  $x - y < 0$, $x - y = 0$, and $x - y > 0$. The \tr{Propagate} rule
  then applies to all three subproblems. Similarly, we branch to
  obtain a truth value for $v > 0$ for every $v$ that appears as the
  annotation of an interface atom.
\end{proof}

\section{SMT as IMT}
\label{sec:smt}

In Section~\ref{sec:bct}, we provided a sound and complete
optimization procedure for the combination of ILP and a
stably-infinite theory (Theorems~\ref{thm:sound-unsat},
\ref{thm:sound-sat}, and \ref{thm:complete}). We will now demonstrate
how to deal with propositional structure, so that we can use this
procedure for SAT Modulo \ZT{} problems, where $T$ is stably-infinite.
In essence, our goal is to flatten propositional structure into linear
constraints.

\subsection{Bounding \ZT{} Instances}

As a prerequisite for dealing with propositional structure, we show
how to bound integer terms in quantifier-free formulas while
preserving \ZT{}-satisfiability.  We build upon well-known results for
ILP~\cite{combopt}. Similar ideas have been applied
to~\Z{}~\cite{lauclid}. Our results go beyond the bounds for~\Z{}, in
that we take into account background theories and objective functions.

Our starting point is known bounds for ILP instances of the form
\begin{equation}
  \begin{aligned}
    \text{min  } c\hspace{2pt} & \mathbf{x} \\
    \text{subject to  } A\hspace{2pt} & \mathbf{x} = b \\
    & \mathbf{x} \geq 0
  \end{aligned}
  \label{eq:ilpformorig}
\end{equation}
where $A$, $b$, and $c$ are matrices of integers ($m \times n$, $m
\times 1$, and $1 \times n$ respectively), and $\mathbf{x}$ is an
$n$-vector of integer variables.

\begin{fact}[{\cite[Corollary of Theorem 13.5]{combopt}}]
  \label{lemma:pap}
  If an ILP instance of the form~(\ref{eq:ilpformorig}) has a finite
  optimum, then it has an optimal solution $x$ such that $x_j \leq
  n^3[(m + 2) d]^{4 m + 12}$ for $1 \leq j \leq n$, where $d =
  \max(\max_{i, j} |A_{i,j}|, \max_i |b_i|, \max_i |c_i|)$.
\end{fact}

In what follows, it will be more convenient to work with the following
form:
\begin{equation}
  \begin{aligned}
    \text{min  } c\hspace{2pt} & \mathbf{x} \\
    \text{subject to  } A\hspace{2pt} & \mathbf{x} = b \\
    D\hspace{2pt} & \mathbf{x} \leq h
  \end{aligned}
  \label{eq:ilpform}
\end{equation}
where $A$, $D$, $b$, $h$ and $c$ are matrices of integers ($m \times
n$, $k \times n$, $m \times 1$, $k \times 1$ and $1 \times n$
respectively), and $\mathbf{x}$ is an $n$-vector of integer variables.

\begin{lemma}
  If an ILP instance I of the form~(\ref{eq:ilpform}) has a finite
  optimum, then it has an optimal solution $\mathbf{x}$ such that
  $|x_j| \leq (2 n + k)^3[(m + k + 2) d]^{4 m + 4 k + 12}$ for $1 \leq
  j \leq n$, where $d = max(\max_{i, j} |A_{i,j}|, \max_i |b_i|,
  \max_i |c_i|, \max_{i,j} |D_{i,j}|, \max_i |h_i|)$.
  \label{lemma:ilpbounds}
\end{lemma}

\begin{proof}
  We reduce $I$ to an equisatisfiable instance over a vector
  $\mathbf{x'}$ of $2 n$ variables constrained to be non-negative
  ($\mathbf{x'} \geq 0$). We achieve this by replacing each variable
  $x_i$ with $x'_i - x'_{n + i}$. In the resulting matrices $A'$, $D'$
  and $c'$, $x'_i$ appears with the same coefficients as $x_i$; $x'_{n
    + i}$ appears with the coefficients multiplied by $-1$. The
  resulting ILP instance $I'$ has $m + k$ constraints and $2 n$
  variables. We replace inequalities with equalities by introducing
  $k$ slack variables and obtain an equisatisfiable instance $I''$
  over a vector $\mathbf{x''}$ of $2 n + k$ variables (the last $k$ of
  which are the slack variables) and $m + k$ constraints.

  The chain of transformations preserves the maximum absolute
  coefficients. We can translate solutions of $I$ to solutions of
  $I''$ and vice versa via the equalities $x_i = x''_i - x''_{n + i}$,
  for $1 \leq i \leq n$; an optimal assignment in either side
  corresponds to an optimum in the other.  $I''$ has a finite optimum
  because $I$ does. By Lemma~\ref{lemma:pap}, $I''$ has an optimal
  solution $\mathbf{y''}$ such that $y''_j \leq (2 n + k)^3[(m + k +
  2) d]^{4 m + 4 k + 12}$. For the corresponding solution $y$ to $I$,
  we have $|y_i| = |y''_i - y''_{n + 1}| \leq (2 n + k)^3[(m + k + 2)
  d]^{4 m + 4 k + 12}$, which concludes our proof.
\end{proof}

We will say that a term is $\Sigma$-rooted if (at its root) it is an
application of a function symbol from the signature $\Sigma$. Let
$\Sigma_0$ and $\Sigma_1$ be signatures such that $\Sigma_0 \cap
\Sigma_1 = \emptyset$. Given a $\Sigma_0 \cup \Sigma_1$-formula $F$,
we will refer to the $\Sigma_i$-rooted terms that appear directly
under predicate and function symbols from $\Sigma_{1 - i}$ as the
$\Sigma_i$-interface terms in $F$. Interface terms are the ones for
which variable abstraction (Example~\ref{ex:abstraction}) introduces
fresh variables.

Let $\Sigma$ be a signature such that $\SigmaZ{} \cap \Sigma =
\emptyset$, and $F$ be a quantifier-free $\SigmaZT{}$-formula. We
denote by $\intfz(F)$ and $\intfsigma(F)$ the sets of
$\SigmaZ$-interface terms and $\Sigma$-interface terms in $F$, and by
$\intf(F)$ the union $\intfsigma(F) \cup \intfz(F)$. Let $\zatoms(F)$
be the set of atomic formulas in $F$ that are applications of $\leq$;
without loss of generality, we will assume that formulas contain no
arithmetic equalities or other kinds of inequalities. Also, let
$\maxc(F)$ be the maximum absolute value among integer coefficients in
$F$ plus one, and $\zvars(F)$ be the set of variable symbols that
appear directly under predicate and function symbols from
\SigmaZ{}. By $o(M)$ we denote the interpretation of linear expression
$o$ under the first-order model $M$. We finally define $\bounds(F,
\rho) = \{-\rho \leq t \wedge t \leq \rho\ |\ t \in \zvars(F) \cup
\intf(F) \}$, for positive integers $\rho$.

\begin{lemma}
  \label{lemma:ztbounds}
  Let $\Sigma$ be a signature such that $\SigmaZ{} \cap \Sigma =
  \emptyset$, and $T$ be a stably-infinite $\Sigma$-theory.
  Furthermore, let $H$ be a finite set of \SigmaZT{} literals, and $o$
  an objective function. Let
  \begin{align*}
    k~=&~|\zatoms(H)| + |\intf(H)| + |\zvars(H)| - 1, \\
    m~=&~|\intfz(H)| + k, \\
    n~=&~|\zvars(H)| + |\intf(H)| \text{, and} \\
    \rho~=&~(2 n + k)^3 [(m + 2) \maxc(H)]^{4m + 12}~.
  \end{align*}
  If there is a first-order model $M$ such that $M \models H \wedge
  \Z{} \wedge T$ and $M$ is a finite optimum for $H$ with respect to
  $o$ (\ie there is some integer constant $c$ such that $M \models o =
  c$ and there is no model $M'$ such that $M' \models H \wedge \Z{}
  \wedge T$ and $M' \models o < c$), then $\{H\} \cup \bounds(H, \rho)
  \cup \{o = o(M)\}$ is \ZT{}-satisfiable.
\end{lemma}

\begin{proof}
  We perform variable abstraction on $H$ (as demonstrated by
  Example~\ref{ex:abstraction}) and obtain a set of \SigmaZ{} literals
  $L$ and a set of $\Sigma$ literals $U$.  Let $V$ be the set of
  variables shared by $L$ and $U$ ($|V| \leq |\intf(H)| +
  |\zvars(H)|$). Let $E$ be the equivalence relation on $V$ induced by
  $M$. Clearly, $L \cup \alpha(V, E)$ is \Z{}-satisfiable and $U \cup
  \alpha(V, E)$ is $T$-satisfiable.  $L$ contains $|\intfz(H)|$
  equalities, $|\zatoms(H)|$ possibly negated inequalities, and
  $|\zvars(H)| + |\intf(H)|$ variables. We eliminate any negations in
  $L$ by rewriting $\neg (S \leq r)$ to $- S \leq - r - 1$ and obtain
  a set of inequalities $L'$.

  We order the variables $V$ so that $v$ is before $u$ if $M \models v
  < u$. (This is slight abuse of notation; $M$ in fact models an
  inequality between the corresponding terms in $H$.) We obtain a
  sequence of variables $v_o, v_1, \ldots, v_{|V| - 1}$, $v_i \in
  |V|$. Let
  \begin{align*}
    \zeta~=&~\{ v_i - v_{i+1} = 0~|~0 \leq i < |V| - 1 \text{ and
    } M \models v_i = v_{i+1} \}, \text{ and } \\
    \eta~=&~\{ v_i - v_{i+1} < 0~|~0 \leq i < |V| - 1 \text{ and } M
    \models v_i < v_{i+1} \}.
  \end{align*}

  Minimizing $o$ subject to $L' \cup \zeta \cup \eta$ is an ILP
  instance $I$ of the form~(\ref{eq:ilpform}). The maximum absolute
  value among coefficients in the matrices representing $I$ is
  $\maxc(H)$. Lemma~\ref{lemma:ilpbounds} applies to $I$. In the worst
  case, $\zeta$ is empty and $|\eta| = |\intf(H)| + |\zvars(H)| - 1$,
  \ie we have $|\intfz(H)|$ equalities and $k$ inequalities. Thus, $I$
  has an optimal assignment $A$ such that for every variable $v$,
  $|A(v)| \leq \rho$. Therefore, $L' \cup \zeta \cup \eta \cup
  \bounds(L', \rho) \cup \{o = o(M)\}$ is \Z{}-satisfiable. Since
  $\zeta \cup \eta \modelsz \alpha(V, E)$, $L' \modelsz L$, and
  $\bounds(L, \rho) = \bounds(L', \rho)$, $L \cup \bounds(L, \rho) \cup
  \{o = o(M)\} \cup \alpha(V, E)$ is \Z{}-satisfiable. Because it is
  also the case that $U \cup \alpha(V, E)$ is $T$-satisfiable, it
  follows from Theorem~\ref{thm:no} that $L \cup U \cup \bounds(L,
  \rho) \cup \{o = o(M)\}$ is \ZT{}-satisfiable. Thus, $H \cup
  \bounds(H, \rho) \cup \{o = o(M)\}$ is \ZT{}-satisfiable.
\end{proof}

\begin{theorem}
  \label{thm:ztbounds}
  Let $\Sigma$ be a signature such that $\SigmaZ{} \cap \Sigma =
  \emptyset$ and $T$ be a stably-infinite $\Sigma$-theory. Let $F$ be
  a quantifier-free $\SigmaZT{}$-formula and $o$ an objective
  function. Let $k = |\zatoms(F)| + |\intf(F)| + |\zvars(F)| - 1$, $m
  = |\intfz(F)| + k$, and $n = |\intf(F)| + |\zvars(F)|$. Finally, let
  $$
  \rho~=~(2 n + k)^3 [(m + 2) \maxc(F)]^{4m + 12}\text{.}
  $$
  If there is a first-order model $M$ such that $M \models F \wedge
  \Z{} \wedge T$ and $M$ is a finite optimum for $F$ with respect to
  $o$ (\ie there is some integer constant $c$ such that $M \models o =
  c$ and there is no model $M'$ such that $M' \models F \wedge \Z{}
  \wedge T$ and $M' \models o < c$), then $\{F\} \cup \bounds(F, \rho)
  \cup \{o = o(M)\}$ is \ZT{}-satisfiable.
\end{theorem}
\begin{proof}
  $M$ satisfies some literals in $F$, and falsifies the rest. Let $H$
  be the set $\{t\ |\ t \text{ is a literal in } F, M \models t\} \cup
  \{\neg t\ |\ t \text{ is a literal in } F, M \models \neg t\}$. $M$
  is a finite optimum for $H$ with respect to $o$: if it was not, it
  would not be a finite optimum for $F$ which contradicts our
  assumptions. For any first-order model $M'$ such that $M' \models H
  \wedge \Z{} \wedge T$, it will also be the case that $M' \models F
  \wedge \Z{} \wedge T$ (because $M$ and $M'$ assign the same truth
  values to the literals that appear in $F$, and the propositional
  structure does not change). By Lemma~\ref{lemma:ztbounds}, $H \cup
  \bounds(\rho) \cup \{o = o(M)\}$ is \ZT{}-satisfiable, \ie there is
  some model $M'$ that satisfies it (and also satisfies \ZT{}). $M'$
  also satisfies $F \cup \bounds(\rho) \cup \{o = o(M)\}$, which
  concludes our proof.
\end{proof}

Intuitively, given a quantifier-free \SigmaZT{}-formula $F$ and an
objective function $o$, Theorem~\ref{thm:ztbounds} allows us to bound
the integer terms of $F$ while preserving (finite)
optimality.~\footnote{A solver that relies on
  Theorem~\ref{thm:ztbounds} for bounding can detect unboundedness by
  imposing the additional constraint $o < o(M)$, re-computing bounds,
  and solving the resulting instance. If the updated instance is
  satisfiable, the original is unbounded.}

\subsection{Propositional Structure}
\label{sec:smt:prop}

Let $\Sigma$ be a signature such that $\SigmaZ{} \cap \Sigma =
\emptyset$ and $T$ be a stably-infinite $\Sigma$-theory. Throughout
this Section, $F$ will be a quantifier-free $\SigmaZT$-formula and $O$
will be an objective function. We show how to encode $F$ as the
conjunction of a set $C$ of integer linear constraints and a set $I$
of $\Sigma$-interface atoms, while preserving optimality with respect
to $O$. We apply a Tseitin-like algorithm, \ie we recursively
introduce $\{0, 1\}$-constrained variables for subformulas of $F$.

The most interesting part is dealing with predicate symbols from
$\Sigma$ and \SigmaZ{}. For the former we simply introduce annotated
$\Sigma$-interface atoms, \eg $[p(x)]^v$. For $\SigmaZ$, we can assume
that we are only confronted with inequalities of the form $\phi =
(\sum_i\ c_i \cdot v_i \leq r)$, because other relations can be
expressed in terms of $\leq$ and the propositional connectives. Also,
we only have to deal with sums over variable symbols, because variable
abstraction takes care of terms that involve $\Sigma$. We define a
variable $v(\phi)$ such that $v(\phi) \Leftrightarrow \phi$ as
follows. By bounding all variables as per Theorem~\ref{thm:ztbounds},
we compute $m$ and $k$ such that $m < \sum_i c_i \cdot v_i \leq k$
always holds.  The direction $v(\phi) \Rightarrow \phi$ can be
expressed as $\sum_i c_i \cdot v_i \leq r + (k - r) \cdot (1 -
v(\phi))$; for the opposite direction we have $\sum_i c_i \cdot v_i >
r + (m - r) \cdot v(\phi)$.

With atomic formulas taken care of, what remains is propositional
connectives; we encode them by using clauses in the standard
fashion. Clauses appear as part of our collection of ILP constraints:
$\vee_i l_i$ is equivalent to $\sum_i l_i \geq 1$. (For translating a
clause to a linear expression, a negative literal $\neg v_i$ appears
as $1 - v_i$ while a positive literal remains intact.)

Note that the (possibly astronomical) coefficients we compute only
serve the purpose of representing formulas as sets of linear
constraints. Their magnitude does not necessarily have algorithmic
side-effects. In the worst case, the initial continuous relaxation
will be weak, but relaxations will become stronger once we start
branching on the Boolean variables. This is no worse than Lazy SMT,
where linear constraints are only applicable once the SAT core assigns
the corresponding Boolean variables.

\comment{Modern ILP solvers provide an alternative called
  \emph{indicator constraints}~\footnote{\url{http://j.mp/NdkZZl}
    (CPLEX); \url{http://j.mp/NdlmDs} (SCIP)}, \ie natively supported
  constraints of the form $v = 1 \Rightarrow \sum_i\ c_i \cdot v_i
  \leq r$, where $v$ is a $\{0, 1\}$ variable. \bct{} does not
  explicitly provide indicator constraints, in order to stay within
  the standard formulation of ILP. }

\section{Implementation and Experiments}
\label{sec:experiments}

IMT first appeared in the context of architectural synthesis for
aerospace systems~\cite{hmp11}. Our approach combined an ILP solver
with a custom decision procedure for real-time constraints. We
implemented the combination in the \cobasa{} tool. The \cobasa{}
manifestation of IMT predates \bct{}. More recently, we implemented a
\bct{}-based solver, which we call \inez{}.

Our experimental evaluation is twofold. First, we show that an ILP
core is essential for the practicality of our synthesis approach. This
part of the evaluation does not deal with \bct{} in any way, but it
nevertheless provides evidence that IMT enables new
applications. Second, we compare our \bct{} prototype against
Z3~\cite{z3} and MathSAT~\cite{mathsatlia} using benchmarks from the
SMT-LIB.

\subsection{Motivation}

\begin{figure}[t]
\includegraphics[trim=0 14pt 0 18pt]{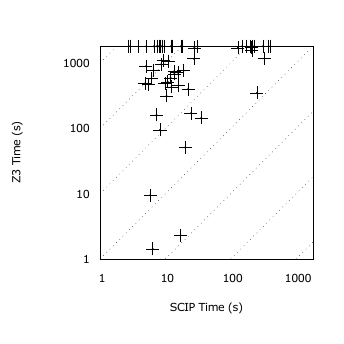}
\includegraphics[trim=0 14pt 0 18pt]{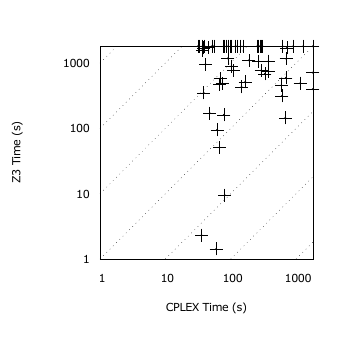}
\caption{Z3 versus SCIP and CPLEX (Synthesis Instances)}
\label{fig:synthesis}
\end{figure}

In the past, we applied \cobasa{} to solve the architectural synthesis
problem for the real, production data from the 787, which was provided
to us by Boeing~\cite{hmp11}. We have made available a family of 60
benchmark instances derived from Boeing problems, 47 of which are
unsatisfiable.~\footnote{\url{http://www.ccs.neu.edu/home/vpap/benchmarks.html}}
We will use these instances to evaluate the suitability of SMT and ILP
solvers as the core of a combination framework for synthesis, which is
a key application area for IMT.

We briefly describe the synthesis problem that gives rise to our
benchmarks. The basic components for this problem are cabinets
(providing resources like CPU time, bandwidth, battery backup,
and memory), software applications (that consume resources), and
global memory spaces (that also consume resources). Applications
and memories have to be mapped to cabinets subject to various
constraints, \eg resource allocation and fault
tolerance. Applications communicate via a publish-subscribe
network. Messages are aggregated into virtual links that are
multicast. The network and messages are subject to various
constraints, \eg bandwidth and scheduling constraints.  The
instances differ in the numbers of different components, the
amounts of different resources, and the collection of structural
and scheduling requirements they encode.

The instances are $\{0, 1\}$-ILP (also known as Pseudo-Boolean). There
are multiple ways to encode $\{0, 1\}$-ILP problems as SMT-LIB
instances. A direct translation led to SMT problems that Z3 could not
solve, so we tried several encodings, most of which yielded similar
results. One encoding was significantly better than the rest, and it
works as follows. Some of the linear constraints are clauses, \ie of
the form $\sum l_i \geq 1$ for literals $l_i$. It makes sense to help
SMT solvers by encoding such constraints as disjunctions of literals
instead of inequalities. To do this, we declare all variables to be
Boolean. Since almost all variables also appear in arithmetic contexts
where they are multiplied by constants greater than 1, we translate
such constraints as demonstrated by the following example: the linear
constraint $x + y + 2 z \geq 2$ becomes \texttt{(>= (+ (ite x 1 0)
  (ite y 1 0) (ite z 2 0)) 2)}.

% Applications communicate via a publish-subscribe network where
% messages are aggregated into virtual links and transmitted via
% multicast. The problem we considered had a large collection of
% applications communicating via thousands of messages. The mapping is
% subject to hard real-time scheduling constraints. The constraints
% are heavily arithmetic due to the resource requirements.

Figure~\ref{fig:synthesis} visualizes the behavior of Z3 versus SCIP
and CPLEX. SCIP solves all instances, while CPLEX solves all but 3. Z3
solves 5 out of 13 satisfiable and 30 out of 47 unsatisfiable
instances. Strictly speaking, the only theory involved is
\Z{}. However, the instances do contain collections of scheduling
theory lemmas~\cite{hmp11} recorded by \cobasa{} in the process of
solving synthesis problems.  Therefore, our setup simulates the kinds
of queries a core solver would be confronted with, when coupled with
our scheduling solver. With suitability for synthesis as the
evaluation criterion, this is the most rigorous comparison we can
perform without implementing and optimizing the combination of SMT
with scheduling. Both ILP solvers significantly outperform Z3,
demonstrating the potential of a general ILP-based combination
framework.

\comment{Recall that with the direct encoding of the ILP problems, Z3
  failed to solve any of the instances above. It was only after we
  experimented with several encodings that we found an encoding that
  allowed Z3 to solve some of the problems. The ILP solvers did not
  require special encodings.}

\subsection{\bct{} Implementation}
\label{ssec:implementation}

\inez{} is implemented as an unobtrusive extension of SCIP.  Namely,
we have extended SCIP with a congruence closure procedure
(\emph{constraint handler} in SCIP terms), and also provide an SMT-LIB
frontend. The overall architecture of SCIP matches \bct{}. Subproblems
(called \emph{nodes}) are created by branching (\tr{Branch}) and
eliminated by operations semantically very similar to \tr{Drop},
\tr{Prune}, \tr{Retire}, and \tr{Unbounded}. SCIP employs various
techniques for cut generation (\tr{Learn}).

Like most modern MIP solvers, SCIP relies heavily on linear
relaxations. While not explicitly mentioned in \bct{}, linear
relaxations fit nicely:
\begin{inparaenum}[(a)]
\item $\operatorname{\mathsf{lb}}$ relies on continuous relaxations,
  as the best integral solution can be at most as good as the best
  non-integral solution.
\item Solutions to relaxations frequently guide branching, \eg if a
  solution assigns a non-integer value $r$ to variable $v$, it makes
  sense to branch around $r$ ($v \geq \lceil r \rceil$ or $v \leq
  \lfloor r \rfloor$).
\item If some relaxation is infeasible, then the corresponding
  subproblem is infeasible and \tr{Drop} applies, while
\item \tr{Retire} or \tr{Unbounded} applies to $T$-consistent integer
  solutions.
\end{inparaenum}

\bct{} proposes difference constraints as a channel of communication
with theory solvers (\tr{Propagate} rule). \inez{} implements
\tr{Propagate} as follows. For every pair of variables $x$ and $y$
whose (dis)equality is of interest to the theory solver, \inez{}
introduces a variable $d_{x,y}$ and imposes the constraint $d_{x,y} =
x - y$. When SCIP fixes the lower bound of $d_{x,y}$ to $l$, the
theory solver is notified of the difference constraint $l \leq x - y$
(similarly for the upper bound). We generally need quadratically many
such auxiliary variables. This is not necessarily a practical issue,
because most pairs of variables are irrelevant.

Our congruence closure procedure takes offsets into
account~\cite{ccoffsets}. In addition to standard propagation based on
congruence closure, \inez{} applies techniques specific to the integer
domain. Notably, if $x$ is bounded between $a$ and $b$, and for every
value of $k$ in $[a, b]$, $f(k)$ is bounded between $l$ and $u$, it
follows that $l \leq f(x) \leq u$, \ie we can impose bounds on
$f(x)$. $a$ and $b$ do not have to be constants, \eg it may be the
case that $m \leq d_{x,y} = x - y \leq n$ and $f(y + m), \ldots, f(y +
n)$ are bounded. We apply this idea dynamically (to benefit from local
bounds) and not just as preprocessing.

\bct{} does not preclude techniques that target special classes of
linear constraints. For example, an implementation can use the
two-watched-literal scheme to accelerate Boolean Constraint
Propagation on clauses.  SCIP implements such techniques. Note that
IMT does not strive to replace propositional reasoning, but rather to
shift a broader class of constraints to the core solver.

\comment{SCIP performs floating-point arithmetic, and can thus provide
  wrong answers. \inez{} inherits this deficiency. For many applications,
  numerical inaccuracies are not a concern, \eg the noise in the model
  overshadows the floating point error intervals. \comment{or an
    answer close enough to the theoretical optimal suffices.}
  Floating-point arithmetic for linear programming is a well-studied
  trade-off. \bct{} can be laid out on top of an exact solver, if the
  application domain dictates accuracy.}

\subsection{\bct{} on SMT-LIB}

\begin{figure}[t]
\includegraphics[trim=0 14pt 0 18pt]{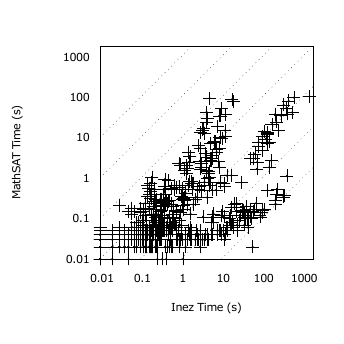}
\includegraphics[trim=0 14pt 0 18pt]{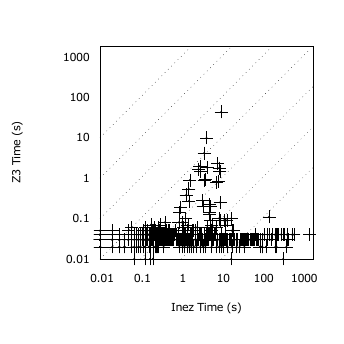}
\caption{\inez{} versus Z3 and MathSAT (SMT-LIB Instances)}
\label{fig:smtlib}
\end{figure}

We experimentally evaluate \inez{} against MathSAT and Z3, based on
the most relevant SMT-LIB category, which is \texttt{QF\_UFLIA}
(Quantifier-Free Linear Integer Arithmetic with Uninterpreted
Functions). Z3 and MathSAT solve all 562 benchmarks, and so does
\inez{}. While \inez{} is generally slower than the more mature SMT
solvers, the majority of the benchmarks (338) require less than a
second, 462 benchmarks require less than 10 seconds, and 528 less than
100 seconds. The integer-specific kind of propagation outlined in
Section~\ref{ssec:implementation} is crucial; we only solve 490
instances with this technique disabled.  Figure~\ref{fig:smtlib}
visualizes our experiments.

Interestingly, the underlying SCIP solver learns no cutting planes
whatsoever for 362 out of the 562 instances. For the remaining
instances the number of cuts is limited. Namely, 126 instances lead to
a single cut, 61 lead to 2 cuts, and the remaining 13 instances lead
to 9 cuts or less. Based on this observation, the branching part of
\inez{}'s branch-and-cut algorithm is being stressed here. We have
not yet tried to optimize branching heuristics, so there is plenty of
room for improvement. More importantly, the instances are not
representative of arithmetic-heavy optimization problems, where we
would expect more cuts.

A final observation is that SCIP performs floating-point (FP)
arithmetic, which may lead to wrong answers. Interestingly, \inez{}
provides no wrong answers for the benchmark set in question, \ie the
instances do not pose numerical difficulties. The fact that we learn
very few cutting planes partially explains why. There is little room
for learning anything at all, let alone for learning something
unsound.

\section{Related Work}
\label{sec:related}

\paragraph{Branch-and-Cut:} 
Branch-and-cut algorithms~\cite{branchandcut-handbook} combine
branch-and-bound with cutting plane techniques, \ie adding violated
inequalities (cuts) to the linear formulation. Different cut
generation methods have been studied for general integer programming
problems, starting with the seminal work of
Gomory~\cite{gomory63}. Cuts can also be generated in a
problem-specific way, \eg for TSP~\cite{gh91}. Problem-specific cuts
are analogous to theory lemmas in IMT.

\paragraph{Nelson-Oppen:} 
The seminal work of Nelson and Oppen~\cite{no79} provided the
foundations for combining decision procedures. Tinelli and
Harandi~\cite{noproof} revisit the Nelson-Oppen method and propose a
non-deterministic variant for non-convex stably-infinite
theories. Manna and Zarba provide a detailed survey of Nelson-Oppen
and related methods~\cite{combiningdp}.\comment{On a fundamental
  level, IMT combines $\Z$ and a theory $T$, so this line of results
  directly applies to our framework.}

\paragraph{SMT:}
ILP Modulo Theories resembles Satisfiability Modulo Theories, with ILP
as the core formalism instead of SAT.  SMT has been the subject of
active research over the last
decade~\cite{barrett02,demoura02,dpllt,z3}. Nieuwenhuis, Oliveras and
Tinelli~\cite{dpllt} present the abstract \dpllt{} framework for
reasoning about lazy SMT. Different fragments of Linear Arithmetic
have been studied as background theories for
SMT~\cite{ladpllt,mathsatlia}. Extensions of SMT support
optimization~\cite{smtopt,smtcosts,laqcost}.

\comment{
  \paragraph{Decomposition:} Another family of linear programming
  techniques that bears resemblance to IMT is Benders
  decomposition~\cite{benders62}. The linear programming problem is
  split into a master problem (which has only a subset of the original
  variables) and a subproblem. The master problem is solved first, and
  then the subproblem is solved with the values of the master problem
  fixed (``trial'' values). If the ``trial'' values are unacceptable,
  a cut is derived and added to the master problem. Logic-based
  Benders Decomposition~\cite{logicbenders} generalizes the strategy
  so that the master problem, the subproblem, or both are not
  necessarily linear programs. In IMT, the problem is ``decomposed''
  into a core ILP instance and background theory problems.  }

\paragraph{Generalized CDCL:} A family of solvers that generalize
CDCL-style search to richer logics recently
emerged~\cite{cutsat,naturalsmt,richerdpll,mcsat}. This research
direction can be viewed as progress towards SMT with a
non-propositional core. Our work is complementary, in the sense that
we do not focus on the core solver, but rather provide a way to
combine a non-Boolean core with theories.

% \subsection{Linear Programming}

\paragraph{Inexact Solvers:} Linear and integer programming solvers
generally perform FP (and thus inexact) calculations. Faure et
al. experiment with the inexact CPLEX solver as a theory
solver~\cite{inexact}\comment{They observe that it does not handle
  well the small, incremental queries that an SMT solver performs,}
and observe wrong answers. For many applications, numerical
inaccuracies are not a concern, \eg the noise in the model overshadows
the floating point error intervals. \comment{or an answer close enough
  to the theoretical optimal suffices.} However, accuracy is often
critical. Recent work~\cite{safeboundsmip,exactrational} proposes
using FP arithmetic as much as possible (especially for solving
continuous relaxations) while preserving safety. \comment{that cutting
  planes, infeasibility certificates, and bounds obtained from
  relaxations are safe.} IMT solvers can be built on top of both exact
and inexact solvers.

\section{Conclusions and Future Work}
\label{sec:conclusions}

We introduced the ILP Modulo Theories (IMT) framework for describing
problems that consist of linear constraints along with background
theory constraints.  We did this via the \bct{} transition system that
captures the essence of branch-and-cut for solving IMT problems. We
showed that \bct{} is a sound and complete optimization procedure for
the combination of ILP with stably-infinite theories. We conducted a
detailed comparison between SMT and IMT.

Many interesting research directions now open up. We could try to
relax requirements on the background theory (\eg stably-infiniteness,
signature disjointness) while preserving soundness and
completeness. We anticipate interesting connections between IMT and
other paradigms, \eg SMT, constraint programming, cut generation, and
decomposition. Additionally, the \bct{} architecture seems to allow
for significant parallelization. Finally, we believe that IMT has the
potential to enable interesting new applications.

\section*{Acknowledgements}

We would like to thank Harsh Raju Chamarthi, Mitesh Jain, and the
anonymous reviewers for their valuable comments and suggestions.

\bibliographystyle{abbrv}
\bibliography{cav_2013}

\begin{thebibliography}{10}

\bibitem{cplex}
{CPLEX}.
\newblock See
  \url{http://www-01.ibm.com/software/integration/optimization/cplex-optimizer/}.

\bibitem{scip}
T.~Achterberg.
\newblock {\em {Constraint Integer Programming}}.
\newblock PhD thesis, Technische Universitat Berlin, 2007.

\bibitem{barrett02}
C.~Barrett, D.~Dill, and A.~Stump.
\newblock {Checking Satisfiability of First-Order Formulas by Incremental
  Translation to SAT}.
\newblock In {\em CAV}, 2002.

\bibitem{unboundedip}
R.~H. Byrd, A.~J. Goldman, and M.~Heller.
\newblock {Recognizing Unbounded Integer Programs}.
\newblock {\em Operations Research}, 35(1):140--142, 1987.

\bibitem{combopt}
{Christos Papadimitriou and Kenneth Steiglitz}.
\newblock {\em {Combinatorial Optimization: Algorithms and Complexity}}.
\newblock Dover, 2nd edition, 1998.

\bibitem{smtcosts}
A.~Cimatti, A.~Franzen, A.~Griggio, R.~Sebastiani, and C.~Stenico.
\newblock {Satisfiability Modulo the Theory of Costs: Foundations and
  Applications}.
\newblock In {\em TACAS}, 2010.

\bibitem{exactrational}
W.~Cook, T.~Koch, D.~Steffy, and K.~Wolter.
\newblock {An Exact Rational Mixed-Integer Programming Solver}.
\newblock In {\em IPCO}, 2011.

\bibitem{z3}
L.~de~Moura and N.~Bjorner.
\newblock {Z3: An Efficient SMT Solver}.
\newblock In {\em TACAS}, 2008.

\bibitem{mcsat}
L.~de~Moura and D.~Jovanovic.
\newblock {A Model-Constructing Satisfiability Calculus}.
\newblock In {\em VMCAI}, 2013.

\bibitem{demoura02}
L.~de~Moura and H.~Ruess.
\newblock {Lemmas on Demand for Satisfiability Solvers}.
\newblock In {\em SAT}, 2002.

\bibitem{ladpllt}
B.~Dutertre and L.~de~Moura.
\newblock {A Fast Linear-Arithmetic Solver for DPLL(T)}.
\newblock In {\em CAV}, 2006.

\bibitem{inexact}
G.~Faure, R.~Nieuwenhuis, A.~Oliveras, and E.~Rodriguez-Carbonell.
\newblock {SAT Modulo the Theory of Linear Arithmetic: Exact, Inexact and
  Commercial Solvers}.
\newblock In {\em SAT}, 2008.

\bibitem{gomory63}
R.~E. Gomory.
\newblock Outline of an algorithm for integer solutions to linear programs.
\newblock {\em Bulletin of the AMS}, 64:275--278, 1958.

\bibitem{mathsatlia}
A.~Griggio.
\newblock {A Practical Approach to Satisfiability Modulo Linear Integer
  Arithmetic}.
\newblock {\em JSAT}, 8:1--27, 2012.

\bibitem{gh91}
M.~Grotschel and O.~Holland.
\newblock {Solution of large-scale symmetric travelling salesman problems}.
\newblock {\em Mathematical Programming}, 51:141--202, 1991.

\bibitem{hmp11}
C.~Hang, P.~Manolios, and V.~Papavasileiou.
\newblock {Synthesizing Cyber-Physical Architectural Models with Real-Time
  Constraints}.
\newblock In {\em CAV}, 2011.

\bibitem{cutsat}
D.~Jovanovic and L.~de~Moura.
\newblock {Cutting to the Chase: Solving Linear Integer Arithmetic}.
\newblock In {\em CADE}, 2011.

\bibitem{richerdpll}
A.~Kuehlmann, K.~McMillan, and M.~Sagiv.
\newblock {Generalizing DPLL to Richer Logics}.
\newblock In {\em CAV}, 2009.

\bibitem{combiningdp}
Z.~Manna and C.~Zarba.
\newblock {Combining Decision Procedures}.
\newblock In {\em {10th Anniversary Colloquium of UNU/IIST}}, 2002.

\bibitem{mc62}
J.~McCarthy.
\newblock {Towards a Mathematical Science of Computation}.
\newblock In {\em Congress IFIP-62}, 1962.

\bibitem{branchandcut-handbook}
J.~E. Mitchell.
\newblock {Branch-and-Cut Algorithms for Combinatorial Optimization Problems}.
\newblock In {\em Handbook of Applied Optimization}, pages 223--233. Oxford
  University Press, 2000.

\bibitem{no79}
G.~Nelson and D.~C. Oppen.
\newblock {Simplification by Cooperating Decision Procedures}.
\newblock {\em TOPLAS}, 1:245--257, 1979.

\bibitem{safeboundsmip}
A.~Neumaier and O.~Shcherbina.
\newblock {Safe bounds in linear and mixed-integer linear programming}.
\newblock {\em Mathematical Programming}, 99:283--296, 2004.

\bibitem{ccoffsets}
R.~Nieuwenhuis and A.~Oliveras.
\newblock {Congruence Closure with Integer Offsets}.
\newblock In {\em LPAR}, 2003.

\bibitem{smtopt}
R.~Nieuwenhuis and A.~Oliveras.
\newblock {On SAT Modulo Theories and Optimization Problems}.
\newblock In {\em SAT}, 2006.

\bibitem{dpllt}
R.~Nieuwenhuis, A.~Oliveras, and C.~Tinelli.
\newblock {Solving SAT and SAT Modulo Theories: From an abstract
  Davis--Putnam--Logemann--Loveland procedure to DPLL(T)}.
\newblock {\em JACM}, 53(6):937--977, 2006.

\bibitem{naturalsmt}
{Scott Cotton}.
\newblock {Natural domain SMT: a preliminary assessment}.
\newblock In {\em FORMATS}, 2010.

\bibitem{laqcost}
R.~Sebastiani and S.~Tomasi.
\newblock {Optimization in SMT with LA(Q) Cost Functions}.
\newblock In {\em IJCAR}, 2012.

\bibitem{lauclid}
S.~A. Seshia and R.~E. Bryant.
\newblock {Deciding Quantifier-Free Presburger Formulas Using Parameterized
  Solution Bounds}.
\newblock In {\em LICS}, 2004.

\bibitem{noproof}
C.~Tinelli and M.~Harandi.
\newblock {A New Correctness Proof of the Nelson-Oppen Combination Procedure}.
\newblock In {\em FroCoS}, 1996.

\end{thebibliography}

\end{document}